\documentclass[conference,letterpaper]{IEEEtran}
\usepackage[utf8]{inputenc} 
\usepackage[numbers,compress]{natbib}
\usepackage[cmex10]{amsmath}
\usepackage{amsfonts,mathrsfs,mathdots,mathtools,amsthm,amssymb,centernot}
\usepackage[bb=ams]{mathalfa}
\usepackage{algorithmic}
\usepackage{graphicx}
\usepackage{textcomp}
\usepackage[dvipsnames]{xcolor}
\usepackage{tikz}
\usepackage{pgfplots}
\usepackage{subcaption}
\usepackage{dsfont}
\usepackage{pifont}
\usepackage{siunitx}
\usepackage{comment}
\usepackage{pgfplotstable}
\usepackage[ruled]{algorithm2e}
\usepackage{color,soul}
\usepackage{bm}	
\usepackage[inline]{enumitem}
\setlist[enumerate,1]{label = \arabic*.,ref = \arabic*}
\usepackage{nicefrac}
\usepackage{subcaption}
\usepackage{colortbl}
\usepackage{booktabs}
\usepackage{diagbox}
\usepackage{url}
\usepackage{ifthen}
\usepackage{balance}
\usepackage{xfrac}
\usepackage{dirtytalk}
\usepackage{cuted}
\usepackage{bbm}

\usepackage{hyperref}  
\hypersetup{colorlinks,
            linkcolor=blue,
            citecolor=blue,
            urlcolor=blue,
            anchorcolor=blue,
            filecolor=blue,
            linktocpage,
            plainpages=false,
            breaklinks=true}

\def\supp{\textnormal{supp}}

\theoremstyle{definition}
\newtheorem{definition}{Definition}

\newtheorem{theorem}[definition]{Theorem}

\newtheorem{corollary}[definition]{Corollary}
\newtheorem{proposition}[definition]{Proposition}
\newtheorem{example}[definition]{Example}
\newtheorem{remark}[definition]{Remark}

\mathtoolsset{showonlyrefs}
\allowdisplaybreaks

\IEEEoverridecommandlockouts
\begin{document}
\title{Bounds on the Privacy Amplification of Arbitrary Channels via the Contraction of $f_\alpha$-Divergence} 

 \author{%
   \IEEEauthorblockN{Leonhard Grosse\IEEEauthorrefmark{1},
                     Sara Saeidian\IEEEauthorrefmark{1}\IEEEauthorrefmark{2},
                     Tobias J. Oechtering\IEEEauthorrefmark{1},
                      and Mikael Skoglund\IEEEauthorrefmark{1}%
   \IEEEauthorblockA{\IEEEauthorrefmark{1}%
                    KTH Royal Institute of Technology, Stockholm, Sweden,
                     \{lgrosse, saeidian, oech, skoglund\}@kth.se}
       \IEEEauthorblockA{\IEEEauthorrefmark{2}%
                    Inria Saclay, Palaiseau, France}}

\thanks{This work was supported by the Swedish Research Council (VR) under grants 2023-04787 and 2024-06615.} 
}

\maketitle

\begin{abstract}
We examine the privacy amplification of channels that do not necessarily satisfy any LDP guarantee by analyzing their contraction behavior in terms of $f_\alpha$-divergence, an $f$-divergence related to Rényi-divergence via a monotonic transformation. We present bounds on contraction for restricted sets of prior distributions via $f$-divergence inequalities and present an improved Pinsker's inequality for $f_\alpha$-divergence based on the joint range technique by \citet{5773031}. The presented bound is tight whenever the value of the total variation distance is larger than $\nicefrac{1}{\alpha}$. By applying these inequalities in a cross-channel setting, we arrive at strong data processing inequalities for $f_\alpha$-divergence that can be adapted to use-case specific restrictions of input distributions and channel. The application of these results to privacy amplification shows that even very sparse channels can lead to significant privacy amplification when used as a post-processing step after local differentially private mechanisms.
\end{abstract}


\section{Introduction}

An important property of \textit{local differential privacy} (LDP) is its invariance to post-processing. Informally, the post-processing property states that whatever further processing is applied to the output of a system satisfying LDP will not weaken the overall system's privacy guarantee. Mathematically, this property can be seen as a form of \emph{data processing inequality} (DPI). As an extension of this observation, it is natural to ask about the \emph{privacy amplification} of a post-processing operation. Basic results state that if a post-processing operation satisfying $\varepsilon_2$-LDP is applied to a system satisfying $\varepsilon_1$-LDP, the composition of both systems will satisfy \emph{at least} $\min\{\varepsilon_1,\varepsilon_2\}$-LDP. However, it is generally not know when \emph{arbitrary} post-processing operations, that is, operations that themselves do \emph{not} satisfy any LDP guarantee, may improve the overall systems LDP guarantee. We take steps to answer this question in this paper. In the context of data processing inequalities, this relates to the investigation of \emph{strong} data processing inequalities (SDPIs), or the \emph{contraction} of a channel with respect to a suitable measure of divergence (see \cite[Chapter 33]{Polyanskiy_Wu_2025}). For pure LDP, the divergence measure of interest is the Rényi-divergence \cite{renyi1961entropy,van2014renyi} of order $\infty$, since
\begin{equation}
    \max_y\log \frac{P_{X|W=w}(y)}{P_{X|W=w'}(y)} = \lim_{\alpha \to \infty} D_\alpha(P_{X|W=w}||P_{X|W=w'}).
\end{equation}
\citet{gilani2024unifying} suggest the notion of $\alpha$-Rényi-LDP as a relaxation of pure LDP, which bounds the maximum Rényi-divergence between the induced output distributions for some arbitrary order $\alpha$. With this framework at hand, studying the privacy amplification of a channel becomes equivalent to studying SDPIs for Rényi-divergences. While SDPIs have been extensively studied for \emph{$f$-divergences}, results on Rényi-divergences are so far less explored. Although Rényi-divergence is not an $f$-divergence, it can be expressed as a monotonic transformation of a corresponding \emph{$f_\alpha$-divergence} \cite{7552457}, a family of $f$-divergences parameterized by $\alpha$. This connection allows us to translate results on $f$-divergences to Rényi-divergences—results that are often easier to obtain.

\subsection{Contributions}

In this work, we derive several technical tools for studying the privacy amplification of arbitrary discrete channels—those not necessarily satisfying any LDP guarantee—by analyzing how they affect the Rényi-LDP of composed systems. Our key contributions are:

\begin{itemize}
    \item We derive a general condition under which channels fail to contract Rényi-divergence, establishing structural limits on the contraction of a channel in Section \ref{sec:generalcond}.
    
    \item We develop new tools for bounding the $f_\alpha$-divergence in terms of total variation distance in Section \ref{sec:pinsker}. Specifically, we provide:
    \begin{itemize}
        \item an improved and often tight version of Pinsker’s inequality for $f_\alpha$-divergence (for $\alpha > 1$);
        \item a simplification of the recently presented proof of Binette's inequality \cite{8630660} in \cite{Hirche:2023caq} based on $E_\gamma$-divergence.
    \end{itemize}
    
    \item Using these tools, we analyze the contraction of $f_\alpha$-divergence and show how input distribution constraints can yield meaningful bounds in Section \ref{sec:sdpi}.
    
    \item Finally, in Section \ref{sec:RLDPex}, we apply these bounds to the Rényi-LDP amplification of arbitrary post-processing channels. Our results show that many channels---despite not satisfying \emph{any} (finite) LDP guarantee---can result in significant privacy amplification if applied to the output of a mechanism satisfying LDP. We illustrate these findings by applying the proposed bounds to randomized response mechanisms post-processed by channels with sparse channel matrices (i.e., matrices with many zeros).
\end{itemize}

\subsection{Related Work}

The contraction behavior of $f$-divergences has been extensively studied; see, e.g., \cite{Polyanskiy_Wu_2025, raginsky2016strong, 10.1134/S0032946020020015, 10206578, anantharam2013maximal, ahlswede1976spreading}. For Rényi-divergences, \citet{polyanskiy2015dissipation} provide upper bounds in the Gaussian setting using the Hockeystick-divergence $E_\gamma$. Recently, \citet{10619367} show that the standard total-variation-based upper bound fails for Rényi-divergence, though the universal $\chi^2$-based lower bound remains valid.
In the context of privacy, \citet{10206578} show that if a mechanism satisfies local differential privacy, strong data processing inequalities can be derived in terms of a channel's LDP guarantee. These bounds are further applied to mixing times in \cite{zamanlooy2024mathrm}. Related, it is shown in \cite{9517999} that the LDP guarantee of a mechanism is closely related to its $f$-divergence contraction properties. The $f_\alpha$-divergence framework enables the use of $f$-divergence inequalities in the Rényi setting, as developed in \cite{7552457} and subsequent works, including an applications to SDPIs in \cite{sason2019data}. The privacy amplification of various mechanisms in the central RDP setting is studied in the line of work \cite{balle2018privacy,balle2019privacy,balle2020privacy}.

\section{Preliminaries}
We use uppercase letters to denote random variables, lowercase letters to denote their realizations and caligraphic letters to denote sets. Specifically, we consider the two random variables $X$ and $Y$, where $X$ represents input data into a channel $P_{Y|X}$ (also referred to as \emph{mechanism} if it satisfies an RLDP guarnatee in the following) that induces a random variable $Y$ at its output. Let $P_{XY}$ denote the joint distribution of $X$ and $Y$. Then we use $P_{XY} = P_{Y|X} \times P_X$ to imply that $P_{XY}(x,y) = P_{Y|X=x}(y)P_X(x)$ and $P_{Y} = P_{Y|X} \circ P_X$ to denote the marginalization $P_Y(y) = \sum_{x \in \mathcal X} P_{Y|X=x}(y)P_X(x)$. We write $\mathcal P_{\mathcal X}$ for the probability simplex on the set $\mathcal X$. For $N\in\mathbb N$, we define $[N] \coloneqq \{1,\dots,N\}$ as the set of positive integers up to $N$. Finally, $\log(\cdot)$ denotes the natural logarithm.

\subsection{Rényi-divergence, $f_\alpha$-divergene and $E_\gamma$-divergence}
For any convex function $f: (0,\infty) \to \mathbb R$ such that $f(1)=0$, define the \emph{$f$-divergence} between two probability measures $P \ll Q$ \cite{csiszar1967information} as
\begin{equation}
    D_f(P||Q) = \mathbb E_Q\bigg[f\bigg(\frac{dP}{dQ}\bigg)\bigg].
\end{equation}
 The \emph{Rényi-divergence} \cite{renyi1961entropy} $D_\alpha(P||Q)$ between $P$ and $Q$ is a monotone transform of the \emph{$f_\alpha$-divergence} \cite{Polyanskiy_Wu_2025}:

\begin{definition}[$f_\alpha$-divergence]
    For any $\alpha > 0$, the \emph{$f_\alpha$-divergence} is defined as the $f$-divergence $D_{f_\alpha}(P||Q)$, where
    \begin{equation}
        f_\alpha(x) = \begin{cases}
            1-x^{\alpha}, \text{ if } \alpha < 1,\\
            x\log x, \text{ if } \alpha = 1, \\
            x^{\alpha}-1, \text{ if } \alpha > 1.
        \end{cases}
    \end{equation}
The \emph{Rényi-divergence of order $\alpha$} is then obtained by
\begin{equation}
\label{eq:falphatoalpha}
        D_\alpha(P||Q) = \frac{1}{\alpha-1} \log\biggl\{D_{f_\alpha}(P||Q) + 1\biggr\}.
\end{equation}
\end{definition}

 We will use the \emph{Hockeystick-divergence} $E_{\gamma}$ (from here on simply called \emph{$E_\gamma$-divergence}) as a tool for obtaining bounds on the Rényi-divergence, as similarly done in \cite[Section 3.2]{polyanskiy2015dissipation}.
\begin{definition}[$E_{\gamma}$-divergence, see, e.g., {{\cite{7552457}}}]
    Given some parameter $\gamma > 0$, the \emph{$E_{\gamma}$-divergence} between two discrete distributions $P$ and $Q$ is defined as
    \begin{equation}
        E_\gamma(P||Q) = \frac{1}{2} \sum_{x \in \mathcal X} |P(x) - \gamma Q(x)| - \frac{1}{2}|1-\gamma|.
    \end{equation}
\end{definition}
For twice differentiable $f$, the relation shown in \cite{cohen1998comparisons},
\begin{equation}
\label{eq:fdivisintegraloverEgamma}
    D_f(P||Q) = \int_0^{\infty}E_{\gamma}(P||Q)f''(\gamma)d\gamma,
\end{equation}
can then be used to obtain any such $f$-divergence. 
For $\gamma = 1$, the $E_\gamma$-divergence is equal to the total variation distance between the two distributions, i.e., $E_1(P||Q) = \text{TV}(P||Q)$.

\subsection{Strong Data Processing Inequalities}
All $f$-divergences (including the $f_\alpha$-divergence)---and hence also the Rényi-divergence---satisfy the \emph{data processing inequality} (DPI). The DPI states that for any two probability measures $P_X, Q_X$ on $\mathcal X$, and some transition kernel $P_{Y|X}$,
\begin{equation}
\label{eq:DPI}
    D_f(P_{Y|X} \circ P_X||P_{Y|X} \circ Q_X) \leq D_f (P_X || Q_X).
\end{equation}We define the contraction coefficient $\eta_{\alpha}(P_{Y|X},\mathcal P)$ of the $\alpha$-Rényi-divergence given a channel $P_{Y|X}$ as the maximum factor between the two sides of the above inequality, that is,
\begin{align}
    &\eta_\alpha(P_X,P_{Y|X},\mathcal P) \\\coloneqq &\sup_{Q_X \in \mathcal P: D_\alpha (P_X||Q_X) \neq 0} \frac{D_\alpha(P_{Y|X} \circ P_X||P_{Y|X} \circ Q_X)}{D_\alpha (P_X || Q_X)},
\end{align}
\begin{equation}
\label{eq:contractioncoeffdef}
    \eta_\alpha (P_{Y|X},\mathcal P) \coloneqq \sup_{P_X \in \mathcal P} \eta_\alpha (P_X,P_{Y|X},\mathcal P).
\end{equation}
Here, $\mathcal P \subseteq \mathcal P_{\mathcal X}$ denotes a subset of the probability simplex on the alphabet $\mathcal X$. In this framework, the DPI states that $\eta_\alpha(P_{Y|X},\mathcal P_{\mathcal X}) \leq 1$. We are then interested in the conditions for which $\eta_{\alpha}(P_{Y|X},\mathcal P) < 1$, that is, cases in which the inequality in \eqref{eq:DPI} is \emph{strict}. In this case, we say a \emph{strong data processing inequality} (SDPI) holds. Similarly to the above, we define $\eta_{\text{TV}}(P_{Y|X},\mathcal P)$ as the contraction coefficient of the total variation distance given that input distributions are picked from the set $\mathcal P$. 
In the most general case, an SDPI considers $\mathcal P = \mathcal P_{\mathcal X}$. However, we will see below that without any restriction, $\eta_{\alpha}(P_{Y|X},\mathcal P_{\mathcal X}) = 1$ in many cases. We will therefore derive bounds for cases in which $\mathcal P \subsetneq \mathcal P_{\mathcal X}$.

\section{Conditions for $\eta_\alpha(P_{Y|X},\mathcal P_{\mathcal X})=1$}
\label{sec:generalcond}
The following result provides a condition on the channel $P_{Y|X}$ for which there is no contraction, that is, for which no strong data processing inequality holds, in the case $\mathcal P = \mathcal P_{\mathcal X}$.
\begin{proposition}
\label{prop:etaisone}
    Consider a discrete channel $P_{Y|X}$. Whenever
    \begin{equation}
    \label{eq:conditionnocontraction}
        \exists x,x' \in \mathcal X: \supp\big(P_{Y|X=x}\big) \cap \supp\big(P_{Y|X=x'}\big) = \emptyset,
    \end{equation}
    we have $\eta_{\alpha}(P_{Y|X},\mathcal P_{\mathcal X}) = 1$.  
\end{proposition}
\begin{IEEEproof}
 We provide a constructive proof in Appendix \ref{app:proofetaisone}. Here, we show Proposition \ref{prop:etaisone} by a simple chain of arguments following from previous results. Let $(X, Y) \sim P_{Y|X} \times P_X$ be arbitrary but fixed. Define the \emph{maximal correlation} as
    \begin{equation}
       \rho_m(X;Y) \coloneqq \sup_{f,g} \mathbb E_{(X,Y) \sim P_{XY}}\Big[f(X)g(Y)\Big],
    \end{equation}
    where the supremum is over all $f: \mathcal X \to \mathbb R$ and $g: \mathcal Y \to \mathbb R$ such that 
    \begin{equation}
        \mathbb E\big[f(X)\big] = \mathbb E\big[g(Y)\big] = 0, \quad \mathbb E\big[f^2(X)\big] = \mathbb E\big[g^2(Y)\big] = 1. 
    \end{equation}
    The contraction coefficient of the $\chi^2$-divergence is closely related to $\rho_m$, and it is shown in, e.g., \cite[Thm. 33.6(c)]{Polyanskiy_Wu_2025} that
    \begin{equation}
        \eta_{\chi^2}(P_{Y|X},\mathcal P_{\mathcal X}) = \sup_{P_X \in\mathcal P_{\mathcal X}} \rho_m^2(X;Y) \eqqcolon s(P_{Y|X}). 
    \end{equation}
    It was shown by \citet{ahlswede1976spreading} that $s(P_{Y|X}) < 1$ if and only if for every pair $x,x'$, it holds that $\supp(P_{Y|X=x}) \cap \supp(P_{Y|X=x'}) \neq \emptyset$. Together with the data processing inequality for $f$-divergences, this implies that whenever \eqref{eq:conditionnocontraction} holds, we have $\eta_{\chi^2}(P_{Y|X},\mathcal P_{\mathcal X}) = 1$. Next, it was shown in \cite[Corollary 2]{10619367} that 
    \begin{equation}
        \eta_\alpha(P_{Y|X},\mathcal P_{\mathcal X}) \geq \eta_{\chi^2}(P_{Y|X},\mathcal P_{\mathcal X}).
    \end{equation}
    By the data processing inequality for Rényi-divergences, we also have $\eta_\alpha(P_{Y|X},\mathcal P_{\mathcal X}) \leq 1$. Hence, if \eqref{eq:conditionnocontraction} holds, we have $\eta_\alpha(P_{Y|X},\mathcal P_{\mathcal X}) = 1$.
\end{IEEEproof}

\subsection{Connection to the Confusion Graph of $P_{Y|X}$}
Proposition~\ref{prop:etaisone} states that a channel does not contract the Rényi-divergence whenever there exist two input symbols that cause conditional distributions with disjoint support at the channel output. This condition can equivalently be expressed as requiring a channel's \emph{confusion graph} to be incomplete \cite{diestel2017}: Define the confusion graph $G=(V,W)$ by $V = \mathcal X$ and
\begin{align}
    W = \{(x_i,x_j) \in \mathcal X^2 \mid \exists\, y\in \mathcal Y: x_i \in \mathcal X^{(y)} \wedge x_j \in \mathcal X^{(y)} \},
\end{align}
where $\mathcal X^{(y)}$ for each $y$ denotes the positions at which the channel matrix contains non-zero entries, that is,
\begin{equation}
        \mathcal{X}^{(y)} \coloneqq \{x \in \mathcal X: P_{Y|X=x}(y) > 0\}.
\end{equation}
Heuristically, an edge $(x_i,x_j)$ between input symbols in the confusion graph indicates that there exists an output symbol $y \in \mathcal Y$ that can be caused by both $x_i$ or $x_j$ at the channel input. That is, a missing edge (compared to the complete version) indicates the existence of a pair of symbols that can be sent over the channel without being confused with each other. The confusion graph---along with its graph-theoretic dual---is extensively used in zero-error information theory, see e.g., \cite{1056798,720537} for a detailed account of results and methods.

Note that a fully connected confusion graph is not necessarily equivalent to a strictly positive channel matrix. Proposition~\ref{prop:etaisone} in a loose sense states that for \emph{sufficiently sparse} $P_{Y|X}$, Rényi-divergences do not contract in general. The following examples illustrates this condition on the channel matrices.
\begin{example}[Block-diagonal matrices]
\label{ex:blockdiagonal}
    Any block-diagonal matrix $\text{diag}(B_1,\dots,B_m)$ with each $B_i \in \mathbb R^{k_i\times k_i},\, k_i \in \mathbb N$, satisfies $\eta_\alpha(P_{Y|X},\mathcal P_{\mathcal X}) = 1$.
\end{example}
\begin{example}[$k$-singular channels]
    Inspired by \cite{6870478}, we say a $N\times N$ channel $P_{Y|X}$ is \emph{$k$-singular} if
    \begin{equation}
        P_{Y|X=x}(y) = \begin{cases}
            \frac{1}{k}, \text{ if } x \in \mathcal X^{(y)} \\
            0, \text{ otherwise.}
        \end{cases}
    \end{equation} By Example~\ref{ex:blockdiagonal}, any $k$-singular channel with $k \leq N/2$ satisfies $\eta_\alpha(P_{Y|X},\mathcal P_{\mathcal X}) = 1$. For illustration, consider the channels
    \begin{equation}
        P = \begin{bmatrix}
            \nicefrac{1}{3} & \nicefrac{1}{3} & \nicefrac{1}{3} & 0\\
            \nicefrac{1}{3} & \nicefrac{1}{3} & 0 & \nicefrac{1}{3}\\
            \nicefrac{1}{3} & 0 & \nicefrac{1}{3} & \nicefrac{1}{3} \\
            0 & \nicefrac{1}{3} & \nicefrac{1}{3} & \nicefrac{1}{3}
        \end{bmatrix}, \quad
        Q = \begin{bmatrix}
            \nicefrac{1}{2} & \nicefrac{1}{2} & 0 & 0 \\
            \nicefrac{1}{2} & \nicefrac{1}{2} & 0 & 0 \\
            0 & 0 & \nicefrac{1}{2} & \nicefrac{1}{2} \\
            0 & 0 & \nicefrac{1}{2} & \nicefrac{1}{2}.
        \end{bmatrix}
    \end{equation}
    Here, $P$ is $3$-singular and $Q$ is $2$-singular. It is easy to see that the assertion in Proposition~\ref{prop:etaisone} applies to $Q$, but not to $P$.
\end{example}

\section{Pinsker-type Inequalities for $f_\alpha$-divergence}
\label{sec:pinsker}
The result in Section \ref{sec:generalcond} shows that in many cases, the contraction coefficient can not be smaller than one whenever $\mathcal P = \mathcal P_{\mathcal X}$. This motivates us to investigate the behavior of $\eta_{\alpha}$ in cases where $\mathcal P \subsetneq \mathcal P_{\mathcal X}$. As we will see in Section \ref{sec:RLDPex}, this can be useful in cases in which partial knowledge about the input distributions to a channel is available.

In this section, we present a generalization of the reverse Pinker's inequality and Pinsker's inequality for Rényi-divergences. These inequalities will then be used to derive \say{SDPI-style} bounds in Section \ref{sec:sdpi}. In most of what follows, we will work with the $f_\alpha$-divergence as a substitute for Rényi-divergence. Due to the monotonic transformation in \eqref{eq:falphatoalpha}, conversion between the two can be done easily at any point.

The following theorem bounds the $f_\alpha$-divergence as a function of the total variation between two distributions. It was first shown by \citet{8630660} for general $f$-divergences. Here, we present a proof that utilizes the integral representation of $f$-divergences via the $E_\gamma$-divergence in \eqref{eq:fdivisintegraloverEgamma}. We remark that the same technique was recently used to proof Binette's inequality in a quantum context by \citet[Proposition 5.2]{Hirche:2023caq}. However, the proof presented below slightly simplifies arguments, and we include it here as we believe it to be instructive. 
\begin{theorem}[Reverse Pinkser's Inequality for $f_\alpha$-divergence, {special case of {\cite[Theorem 1]{8630660}}}]
\label{thm:gammabound}
Let $\alpha > 1$. For any two distributions $P_X, Q_X \in \mathcal P$, define
\begin{equation}
    \Gamma_{\max}(X) \coloneqq \sup_{P_X, Q_X \in \mathcal P, \, x\in\mathcal X} \frac{P_X(x)}{Q_X(x)},
\end{equation}
\begin{equation}
    \Gamma_{\min}(X) \coloneqq \inf_{P_X,Q_X \in \mathcal P, \, x\in\mathcal X} \frac{P_X(x)}{Q_X(x)}.
\end{equation}
Then, the maximum $f_\alpha$-divergence between the two distributions is upper bounded by
\begin{equation}
    D_{f_\alpha}(P_X||Q_X) \leq \textnormal{TV}(P_X||Q_X)R_\alpha(\Gamma_{\max}(X),\Gamma_{\min}(X)),
\end{equation} 
where for $u \in [1,\infty) \cup \{\infty\}$ and $v \in [0,1]$,
\begin{equation}
    R_\alpha(u,v) = \frac{u^{\alpha}-1}{u-1} - \frac{1-v^{\alpha}}{1-v}.
\end{equation}
\end{theorem}
\begin{IEEEproof}
Fix $Q_X$ and $P_X$ and let $\Delta \coloneqq \text{TV}(P_X||Q_X)$. The proof depends on the fact that the only interval on which we may have $E_\gamma(P_X||Q_X) \neq 0$ is  $(\Gamma_{\min}(X),\Gamma_{\max}(X))$. To see this, observe that 
    \begin{align}
        E_\gamma(P_X||Q_X) = \frac{1}{2}\sum_{x \in \mathcal X}\biggl|P_X(x)-\gamma Q_X(x)\biggr| -\frac{1}{2}|1-\gamma|,
    \end{align}
    and whenever $P_X(x) - \gamma Q_X(x) \leq 0$ for all $x \in \mathcal X$ or  $P_X(x) - \gamma Q_X(x) \geq 0$ for all $x \in \mathcal X$, we have
    \begin{align}
        E_\gamma(P_X||Q_X) &= \frac{1}{2}\sum_{x \in \mathcal X}\biggl|P_X(x)-\gamma Q_X(x)\biggr| -\frac{1}{2}|1-\gamma| \\&= \frac{1}{2}\sum_{x \in \mathcal X}\biggl(\pm P_X(x) \mp \gamma Q_X(x)\biggr) -\frac{1}{2}|1-\gamma| \\&= \frac{1}{2}|1-\gamma| -\frac{1}{2}|1-\gamma| = 0.
    \end{align}
    This happens exactly when 
    \begin{equation}
        \gamma \leq \min_x \frac{P_X(x)}{Q_X(x)} \eqqcolon \Gamma_{\min}(X,P_X,Q_X),  
    \end{equation}
    or,
    \begin{equation}
        \gamma \geq \max_x \frac{P_X(x)}{Q_X(x)} \eqqcolon \Gamma_{\max}(X,P_X,Q_X).
    \end{equation}
     Hence, $E_\gamma$ must be zero everywhere but on the interval between these two quantities, that is, the only values of $\gamma$ that allow a non-zero value of $E_\gamma$ are
    \begin{equation}
        \gamma \in (\Gamma_{\min}(X,P_X,Q_X),\Gamma_{\max}(X,P_X,Q_X)) \eqqcolon I_{\geq 0}.
    \end{equation}
    From this, we get
    \begin{equation}
        E_\gamma(P_X||Q_X) = \begin{cases} 
            0, &\text{ if } \gamma \notin I_{\geq 0}, \\[.5em]
            E_\gamma(P_X||Q_X), &\text{ if } \gamma \in I_{\geq 0},
        \end{cases}
    \end{equation}
     What remains is to bound $E_\gamma$ on the interval $I_{\geq0}$, which yields an upper bound on $D_{f_\alpha}$ by \eqref{eq:fdivisintegraloverEgamma}. To do this, note that $\gamma \mapsto E_\gamma(P||Q)$ is convex and increasing on $[0,1]$ and convex and decreasing on $(1,\infty)$ \cite{polyanskiy2015dissipation}. Hence, the maximum value is obtained at $\gamma = 1$, where we have $E_1(Q_X||P_X) = \Delta$. By convexity, we can then upper bound $E_\gamma$ by the piecewise linear function 
    \begin{align}
        E_\gamma&(P_X||Q_X) \leq g(\gamma) \\&\coloneqq \begin{cases}
            0, &\text{if } \gamma \notin I_{\geq 0},\\[.7em]
            \Delta \frac{\gamma - \Gamma_{\min}(X,P_X,Q_X)}{1-\Gamma_{\min}(X,P_X,Q_X)}, &\text{if } \gamma \in I_{\geq 0} \cap (0,1],\\[.7em]
            \Delta\frac{-\gamma + \Gamma_{\max}(X,P_X,Q_X)}{\Gamma_{\max}(X,P_X,Q_X)-1}, &\text{if } \gamma \in I_{\geq 0} \cap (1,\infty).
        \end{cases}
    \end{align}
    The theorem then follows from maximizing $g(\gamma)$ over $P_X,Q_X \in \mathcal P$, which together with \eqref{eq:fdivisintegraloverEgamma} yields
    \begin{equation}
    \label{eq:renyiEgammaintegral}
        D_{f_\alpha}(P_X||Q_X) \leq \alpha(\alpha-1) \int_{\Gamma_{\min}(X)}^{\Gamma_{\max}(X)}g(\gamma)\gamma^{\alpha-2}d\gamma,
    \end{equation}
    and the desired statement follows by solving this integral.
\end{IEEEproof}
Note that this proof technique can also be used to obtain the more general result in \cite{8630660} for any $f$-divergence for which $f$ is twice differentiable by replacing $\gamma^{\alpha-2}$ with the general $f''(\gamma)$ in \eqref{eq:renyiEgammaintegral}. Further, if more information about the distributions is known, more specific bounds on $f(\gamma)$ can potentially be used to obtain more specific inequalities in these cases.
\begin{remark}
    The bound in Theorem~\ref{thm:gammabound} recovers the definition of Rényi-divergence of order infinity, that is, it is asymptotically tight with $D_\infty(P_X||Q_X) = \log\Gamma_{\max}(X)$.
    To see this, let again $\Delta \coloneqq \text{TV}(P_X||Q_X)$ and observe that we can plug in the result form Theorem~\ref{thm:gammabound} into \eqref{eq:falphatoalpha} to obtain
    \begin{equation}
        D_{\alpha}(Q_X||P_X) \leq \frac{1}{\alpha-1}\log\big(\Delta R_\alpha(\Gamma_{\max}(X),\Gamma_{\min}(X)) + 1\big).
    \end{equation}
    Since $\Gamma_{\min}(X) \leq 1$ and $\Gamma_{\max}(X) \geq 1$, we get the limit as
    \begin{align}
        D_\infty(Q_X||P_X) &\leq \lim_{\alpha \to \infty}\frac{1}{\alpha-1}\log \biggl(R_\alpha\big(\Gamma_{\max}(X),\Gamma_{\min}(X)\big)\bigg) 
        \\&= \lim_{\alpha \to \infty}\frac{\alpha}{\alpha-1}\log \Gamma_{\max}(X) = \log \Gamma_{\max}(X).
    \end{align}
    The bound in \cite[Theorem 35]{7552457}, \cite[Theorem 3]{7360766} is obtained from Theorem~\ref{thm:gammabound} by assuming the lower bound $\Gamma_{\min}(X) = 0$.
\end{remark}

Next, we present a generalized Pinsker's inequality for $f_\alpha$-divergence of order $\alpha > 1$.

\begin{theorem}[Pinsker's inequality for $f_\alpha$-divergences]
\label{thm:pinsker}
    For any $\alpha>1$ and any two mutually absolutely continuous probability measures $P_X \ll \gg Q_X$, we have 
    \begin{equation}
         D_{f_\alpha}(P_X||Q_X) \geq g_\alpha\biggl(\text{TV}(P_X||Q_X)\biggr),
    \end{equation}
    where
    \begin{equation}
    \label{eq:pinskerfunction}
        g_\alpha(t) = \begin{cases}
            e^{2(\alpha-1)t^2}-1, &\text{ if } t < \nicefrac{1}{\alpha} \text{ and }\alpha < 2,\\
            (4t^2+1)^{\alpha-1}-1, &\text{ if } t<\nicefrac{1}{\alpha} \text{ and } \alpha \geq 2,\\
            (1-t)^{1-\alpha}-1, &\text{ otherwise. }
        \end{cases}
    \end{equation}
\end{theorem}
\begin{IEEEproof}
  The proof follows from bounding the \emph{joint range} \cite{5773031} of the divergences and is presented in Appendix \ref{app:proofpinsker}.
\end{IEEEproof}
\begin{remark}
For any $\alpha>1$, the bound in Theorem~\ref{thm:pinsker} is attained by the pair of distributions $P_X = [0,1]$, $Q_X = [t,1-t]$ in the case  $t \geq \nicefrac{1}{\alpha}$. Hence, it is tight for this range of values. It is also tight for all $t\geq 0$ given that $\alpha = 2$. In this case, it reduces to the joint-range-based bound between $\chi^2$-divergence and total variation distance, as e.g. presented in \cite{Polyanskiy_Wu_2025}. As a result of the former, it is also asymptotically tight for $\alpha \to \infty$, since in this case, the condition $t \geq \nicefrac{1}{\alpha}$ holds for all $t > 0$.
    Note a tight bound in the case $t<\nicefrac{1}{\alpha}$ can be obtained by solving
\begin{equation}
    \frac{\partial}{\partial p}\psi_\alpha (p,t) 
       = 0.
\end{equation}
In the range $t\in (0,\nicefrac{1}{\alpha})$, this equation has a real solution $p^*$ that satisfies $p^* \in [0,1]$. Due to the convexity of $\psi_\alpha$ in $p$, a (tight) lower bound is then given by $g_\alpha(t) = \psi_\alpha(p^*,t)$.
However, it remains unclear if $p^*$ can be expressed in closed form. Due to this complication, the result above utilizes more rudimentary bounds in this case. Further, we invoke standard Pinsker's inequality for relative entropy to derive the first case of \eqref{eq:pinskerfunction}. This bound is not the best we can do for this case. More specifically, the optimal Pinsker's inequality has been found in \cite{1201071} in implicit form. We rely on the more rudimentary standard Pinsker bound here for simplicity.
\end{remark}
\begin{remark}
\label{rem:pinskerfuncinv}
    The function $g_\alpha$ is injective, but not bijective. That is, in order to obtain an inverse bound (later used in Corollary~\ref{corr:SDPI}) we can define the piecewise function 
    \begin{equation}
        g^{-1}_\alpha(s) = \begin{cases}
            \sqrt{\frac{\log(s+1)}{2(\alpha-1)}}, &\text{if } \alpha < 2, s <h_1(\alpha), \\
            \frac{1}{2}\sqrt{(s+1)^{(\frac{1}{\alpha-1})}-1}, &\text{if } \alpha \geq 2, s < h_2(\alpha) \\
            \max\{1-(s+1)^{\frac{1}{1-\alpha}},\nicefrac{1}{\alpha}\}, &\text{otherwise},
        \end{cases}
    \end{equation}
    where $h_2(\alpha) \coloneqq (1+4\alpha^{-2})^{\alpha-1}-1$ and $h_1(\alpha) \coloneqq 2-\nicefrac{2}{\alpha}$.
\end{remark}
An example of the presented bound on the joint range for $\alpha=4$ is shown in Figure \ref{fig:jointrange}. Other values of $\alpha$ yield a visually similar result, with the discontinuity at $\nicefrac{1}{\alpha}$ moving closer to the origin as the value of $\alpha$ increases. 
\begin{figure}
\centering
\includegraphics[scale=0.45]{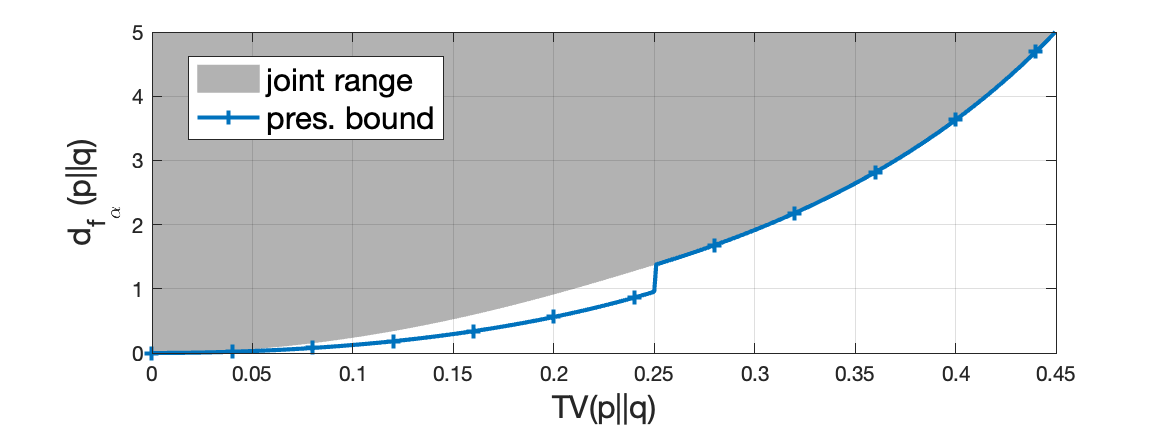}
    \caption{Joint range of $D_{f_{\alpha}}(P||Q)$ and $\text{TV}(P||Q)$ for the case $\alpha=4$ as well as the bound presented in Theorem~\ref{thm:pinsker}.}
    \label{fig:jointrange}
\end{figure}

\section{Strong Data Processing with Pinsker}
\label{sec:sdpi}
The above results can now be used to reason about the data processing properties of channels by considering the presented inequalities \say{across} a given channel. That is, we can modify the bound presented in Theorem~\ref{thm:gammabound} to yield an upper bound on $D_{f_\alpha}(P_Y||Q_Y)$ in terms of $\text{TV}(P_X||Q_X)$. To this end, consider 
\begin{equation}
\label{eq:crosschannelgammamax}
    \Gamma_{\max}(Y)  = \sup_{P_X,Q_X \in \mathcal P,\,y\in \mathcal Y} \frac{\sum_{x\in\mathcal X}P_{Y|X=x}(y)Q_X(x)}{\sum_{x\in\mathcal X}P_{Y|X=x}(y)P_X(x)},
\end{equation}
\begin{equation}
\label{eq:crosschannelgammamin}
    \Gamma_{\min}(Y)  = \inf_{P_X,Q_X \in \mathcal P,\,y\in \mathcal Y} \frac{\sum_{x\in\mathcal X}P_{Y|X=x}(y)Q_X(x)}{\sum_{x\in\mathcal X}P_{Y|X=x}(y)P_X(x)},
\end{equation}
as the \say{cross-channel} modification of the crucial extremal quantities in Theorem~\ref{thm:gammabound}. We get the following result.
\begin{corollary}
\label{corr:SDPI}
    For a discrete channel $P_{Y|X}$, distributions $P_X, Q_X \in \mathcal P \subseteq \mathcal P_{\mathcal X}$, $P_Y = P_{Y|X}\circ P_X$, $Q_Y = P_{Y|X}\circ Q_X$ and some $\alpha > 1$, we have
    \begin{align}
        D_{f_\alpha}&(P_Y||Q_Y) \\&\leq \eta_{TV}(P_{Y|X},\mathcal P) R_{\alpha}(\Gamma_{\max}(Y),\Gamma_{\min}(Y))\\&\qquad \qquad \qquad \qquad \qquad \times g_\alpha^{-1}\Big(D_{f_\alpha}(P_X||Q_X)\Big).
    \end{align}
\end{corollary}
\begin{proof}
    Follows from Theorem~\ref{thm:gammabound} with $P_Y = P_{Y|X} \circ P_X$ and $Q_Y = P_{Y|X} \circ Q_X$, the inequality in Theorem~\ref{thm:pinsker}, the inverse given in Remark \ref{rem:pinskerfuncinv} and since by definition, $\text{TV}(P_Y||Q_Y) \leq \eta_{TV}(P_{Y|X},\mathcal P)\text{TV}(P_X||Q_X)$ if $P_X,Q_X \in \mathcal P$.
\end{proof}
This result stresses the relation between $\Gamma_{\max}(Y)$ and $\Gamma_{\min}(Y)$ as formulated in \eqref{eq:crosschannelgammamax} and \eqref{eq:crosschannelgammamin} and the contraction behavior of a channel in terms of Rényi-divergence: If $\Gamma_{\max}(Y)-\Gamma_{\min}(Y)$ is small, the channel significantly contracts the $f_\alpha$-divergence. If the difference is large, contraction will be small.
What remains is to bound $\Gamma_{\max}(Y)$ and $\Gamma_{\min}(Y)$ so that an operational bound on this contraction is obtained. Such bounds should be use-case specific, as the structure of the uncertainty set $\mathcal P$, as well as specifics about the channel are crucial factors. One such bound can be obtained in the application of the above results to privacy amplification, which we will present in the following section.

\section{Privacy Amplification Bounds for RLDP}
\label{sec:RLDPex}
In this section, we show how the above technical tools can be used to obtain improved post-processing inequalities for Rényi local differential privacy (RLDP) \cite{mironov2017renyi,gilani2024unifying}. Consider the Markov chain $W-X-Y$, where $W$ is a random variable on a set $\mathcal W$ that represents some private data, and the transition kernel $P_{X|W}$ is a mechanism satisfying $(\varepsilon,\alpha)$-RLDP,\footnote{The notion of Rényi \emph{local} differential privacy as a local adaption of Rényi differential privacy \cite{mironov2017renyi} was recently suggested by \citet{gilani2024unifying}, where the authors refer to it as \emph{local Rényi differential privacy} (LRDP). Note that the class of discrete mechanisms satisfying any (finite) RLDP guarantee for any $\alpha$ is equivalent to the class of discrete mechanisms satisfying some (finite) LDP guarantee. Specifically, this class only includes mechanisms with strictly positive entries in the mechanism matrix.} i.e., 
\begin{align}
\text{RLDP}^{(\alpha)}&(P_{X|W})\\&\coloneqq \max_{w,w' \in \mathcal W,\,y \in \mathcal Y}D_{\alpha}(P_{X|W=w}||P_{X|W=w'}) \leq \varepsilon.
\end{align}
Further, assume that the channel $P_{Y|X}$ does not satisfy any finite RLDP guarantee. From the post-processing property of RLDP (which is just the DPI for Rényi-divergence), it is clear that the privacy guarantee of the system $P_{Y|W}$ will remain bounded by the guarantee of $P_{X|W}$. However, it is reasonable to assume that there are certain $P_{Y|X}$ with zero-valued elements that improve the privacy guarantee of $P_{Y|W}$ when cascaded with the RLDP mechanism $P_{X|W}$. Using the results above, we can bound the $f_\alpha$-divergence of $P_{Y|W}$ by
\begin{align}
    D_{f_\alpha}&(P_{Y|W=w}||P_{Y|W=w'}) \\&= D_{f_\alpha}(P_{Y|X} \circ P_{X|W=w}||P_{Y|X} \circ P_{X|W=w'}) \\&\qquad \leq \eta_{TV}(P_{Y|X},\mathcal P_\mathcal X) R_\alpha(\Gamma_{\max}(Y),\Gamma_{\min}(Y))\\&\qquad \qquad \qquad \quad\times g_\alpha^{-1}\Big(D_{f_\alpha}(P_{X|W=w}||P_{X|W=w'})\Big).
\end{align}
Since $P_{X|W}$ satisfies $(\varepsilon,\alpha)$-RLDP, we can apply \eqref{eq:falphatoalpha} to bound the RLDP guarantee of the entire system as 
\begin{equation}
\label{eq:RLDPbound}
    \varepsilon^* \coloneqq \text{RLDP}^{(\alpha)}(P_{Y|W}) \leq \varphi(\varepsilon_f,\alpha),
    \end{equation}
where $\varepsilon_f \coloneqq \max_{w,w'}D_{f_\alpha}(P_{X|W=w}||P_{X|W=w'})$ is the $f_\alpha$-divergence corresponding to the $\varepsilon$-valued Rényi-divergence, and we define $\varphi: [0,\infty)\times (1,\infty) \to \mathbb R$ as
    \begin{align}
        \varphi&(\varepsilon_f,\alpha) \\& \coloneqq \frac{1}{\alpha-1}\log\biggl(\eta_{TV} R_\alpha(\Gamma_{\max}(Y),\Gamma_{\min}(Y))g_\alpha^{-1}(\varepsilon_f)+1\biggr).
    \end{align}

Further, in this setup, the set of input distributions $\mathcal P$ to the channel $P_{Y|X}$ is just the set of rows of the RLDP mechanism $P_{X|W}$, and each of these distributions is induced by one of the outcomes $W = w$. Hence, the quantities $\Gamma_{\max}(Y)$ and $\Gamma_{\min}(Y)$ reduce to
\begin{equation}
    \Gamma_{\max}(Y) = \max_{w,w' \in \mathcal W, \,y\in \mathcal Y} \frac{\sum_{x\in\mathcal X}P_{Y|X=x}(y)P_{X|W=w}(x)}{\sum_{x\in\mathcal X}P_{Y|X=x}(y)P_{X|W=w'}(x)},
\end{equation}
\begin{equation}
    \Gamma_{\min}(Y) = \min_{w,w' \in \mathcal W, \,y\in \mathcal Y} \frac{\sum_{x\in\mathcal X}P_{Y|X=x}(y)P_{X|W=w}(x)}{\sum_{x\in\mathcal X}P_{Y|X=x}(y)P_{X|W=w'}(x)}.
\end{equation}

The following example illustrates the above results. 
    \begin{figure*}[!t]
        \centering
        \begin{subfigure}{.49\textwidth}
            \centering
            \includegraphics[scale=0.525]{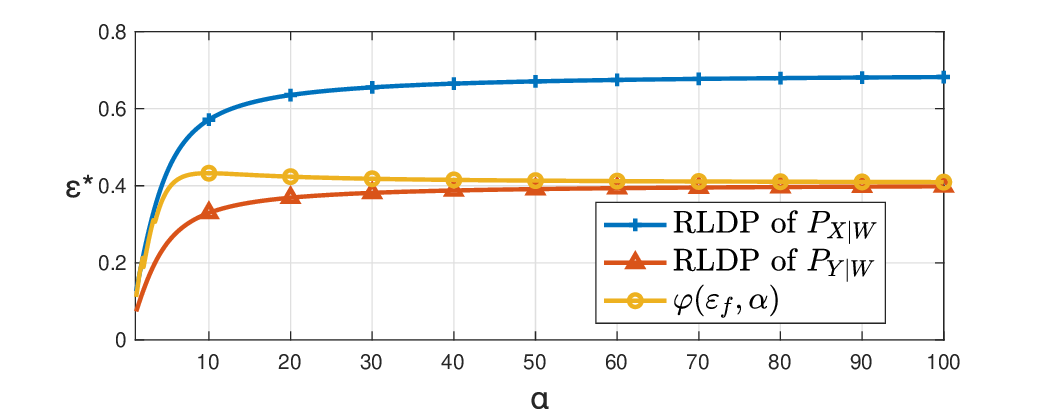}
            \caption{$N=5$, $\varepsilon=\log2$, $k=2$}
        \end{subfigure}
        \begin{subfigure}{.49\textwidth}
            \centering
            \includegraphics[scale=0.525]{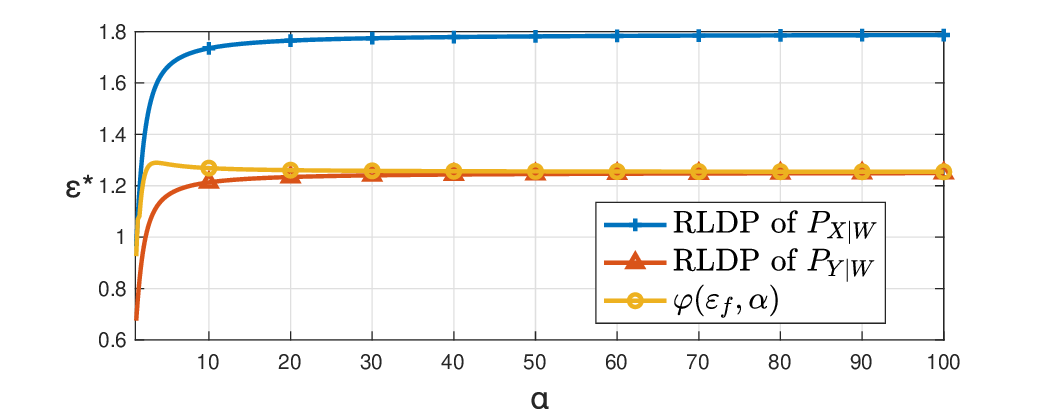}
            \caption{$N=5$, $\varepsilon=\log6$, $k=2$}
        \end{subfigure}
        \begin{subfigure}{.49\textwidth}
            \centering
            \includegraphics[scale=0.525]{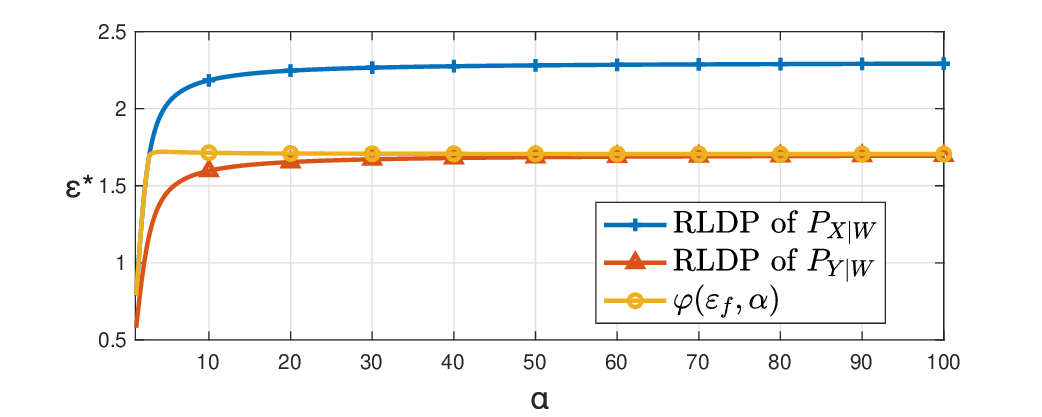}
            \caption{$N=20$, $\varepsilon=\log10$, $k=2$}
        \end{subfigure}
        \begin{subfigure}{.49\textwidth}
            \centering
            \includegraphics[scale=0.525]{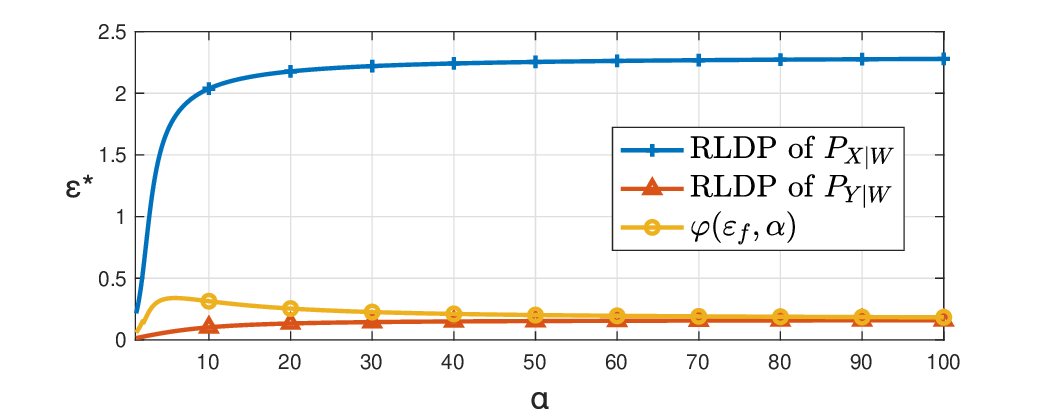}
            \caption{$N=100$, $\varepsilon=\log10$, $k=50$}
        \end{subfigure}
        \caption{RLDP guarantees and the corresponding bounds in Example~\ref{ex:RLDPex}.}
        \label{fig:RLDPex}
    \end{figure*}
\begin{example}
\label{ex:RLDPex}
    Let $P^{(N)}_{X|W}$ be the $N\times N$ randomized response mechanism in \cite{extremalmechanismLong} satisfying $\varepsilon$-LDP, that is, let
    \begin{equation}
        P^{(N)}_{X|W=w}(x) = \begin{cases}
            \frac{e^\varepsilon}{N+e^\varepsilon-1}, &\text{ if }x=w\\
            \frac{1}{N+e^\varepsilon-1}, &\text{ otherwise}.
        \end{cases}
    \end{equation}
For some $m,k > 0$, let $N = mk$. We examine the privacy amplification of the channels $Q^{(m,k)}_{Y|X}$ defined by 
    \begin{equation}
                Q_{Y|X}^{(m,k)} \coloneqq \text{diag}
                (\underbrace{U^{(k)},\dots,U^{(k)}}_{m \text{ times}}),
    \end{equation}
     where $U^{(k)} = \frac{1}{k}\mathbbm 1_k\mathbbm 1_k^T$ is a \say{uniform} block of size $k\times k$.  
    In this setup, the successive application of $P_{X|W}^{(N)}$ and $Q_{Y|X}^{(m,k)}$ to the private data $W$ will increase randomization, and hence it is reasonable to expect an improved privacy guarantee of the composed system compared to $P^{(N)}_{X|W}$ alone.
    Computing the values of $\Gamma_{\max}$ and $\Gamma_{\min}$ in this setup, we find that given the channel $Q^{(m,k)}_{Y|X}$, we have
    \begin{equation}
        \Gamma_{\max}(Y) =\Gamma_{\min}(Y)^{-1}= 1+ \frac{e^\varepsilon-1}{k} = 1 + \frac{m(e^\varepsilon-1)}{N}.
    \end{equation}   
 Interestingly, we observe that for \emph{any} $\varepsilon>0$, and \emph{any} fixed ratio of non-zero elements $\nicefrac{k}{N} = \nicefrac{1}{m}$, we have
 \begin{equation}
      \lim_{N\to \infty}\Gamma_{\max}(Y) = \lim_{N\to \infty}\Gamma_{\min} = 1,
 \end{equation}
 which implies that for any $\alpha >1$, we have
 \begin{equation}
     \varepsilon^* = \text{RLDP}^{(\alpha)}\Big(Q^{(\nicefrac{N}{k},k)}_{Y|X} \circ P_{X|W}^{(N)}\Big) \longrightarrow 0.
 \end{equation}
 That is, post-processing a randomized response mechanism of \emph{any} finite LDP guarantee with a channel constructed in the above way will result in a system with perfect privacy as $N\to \infty$ whenever the ratio of non-zero elements remains fixed as $N$ grows.
We can further numerically determine the bound in \eqref{eq:RLDPbound} assuming the upper bound $\eta_{TV}(P_{Y|X},\mathcal P_{\mathcal X}) = 1$. The resulting RLDP values are shown in Figure \ref{fig:RLDPex} alongside the RLDP guarantee of $P^{(N)}_{X|W}$ (the best bound previously known) and the true RLDP guarantee of the compositions for different choices of $N$ and $k$. As expected, the bound tightens as $\alpha$ grows, since both of the presented Pinsker-type bounds are asymptotically tight in the parameter $\alpha$. 
    \end{example}
Example \ref{ex:RLDPex} shows that even sparse channel matrices can lead to significant privacy amplification if applied as a post-processing step to an LDP system. Moreover, as long as the number of non-zero elements in the channel matrix grows at least linearly with $N$, the privacy guarantee of the system will approach prefect privacy as the alphabet size grows. While these insights are so far limited to randomized response mechanisms, they are interesting for two main reasons: 

First, the results motivate us to further investigate \emph{arbitrary} systems that apply some form of randomization in terms of their effect on LDP guarantees; it has been previously observed that local differential privacy constraints on a system can be severely limiting, especially in large-scale learning applications \cite{10.1145/3433638,10.5555/3454287.3455674}. In this context, the results presented in this paper are encouraging. Specifically, if some form of LDP processing is done to the data, even a learning system without a formal LDP guarantee on its own might be able to significantly increase the privacy guarantee with respect to the data used to train it.

Second, the results are promising for addressing the common issue of large values of $\varepsilon$ in $\varepsilon$-LDP guarantees in practical systems \cite{dwork2019exposeyourepsilons}: Instead of fine tuning the practical application itself to improve its privacy guarantee (which might be intractable), practitioners might be able to apply simple post-processing steps to an application output in order to significantly improve the systems privacy guarantee, while retaining much of the original systems performance.

Future work in this direction should focus on more advanced tools and methods for bounding the values $\Gamma_{\max}$ and $\Gamma_{\min}$ for various (practical) systems. Obtaining such bounds paves the way for easy plug-and-play solutions for improving privacy guarantees of weak LDP systems at a minimal loss of utility.

\normalsize
\appendices
    \section{Proof of Theorem~\ref{thm:pinsker}}
    \label{app:proofpinsker}
      From \cite{7552457} and the fact that $D_{f_{\alpha}}$ is a monotone transform of the Rényi-divergence of the same order, we know that
    \begin{equation}
        \inf_{P,Q\in \mathcal P_X:\, TV(P||Q)=t} D_{f_\alpha}(P||Q)= \min_{p,q: \,|p-q|=t} d_{f_\alpha}(p||q),
    \end{equation}
    where $d_{f_\alpha}(p||q)$ denotes the binary $f_\alpha$-divergence
    \begin{equation}
        d_{f_\alpha}(p||q) \coloneqq \frac{p^\alpha}{q^{\alpha-1}} + \frac{(1-p)^\alpha}{(1-q)^{\alpha-1}} - 1.
    \end{equation}
    Define 
    \begin{equation}
        \psi_{\alpha}(p,t) \coloneqq \frac{p^\alpha}{(p+t)^{\alpha-1}} + \frac{(1-p)^\alpha}{(1-p-t)^{\alpha-1}} - 1.
    \end{equation}
    From the definition of the total variation, we get $\min_{p,q:\,|p-q|=t} d_{f_\alpha}(p||q) = \min_{p\in[0,1-t]} \{\psi_\alpha(p,t)\}$. Taking derivatives with respect to $p$ yields
    \begin{equation}
        \frac{\partial}{\partial p}\psi_\alpha (p,t) = (\alpha t + p)\biggl[\frac{p^{\alpha-1}}{(p+t)^\alpha} + \frac
        {(1-p)^{\alpha-1}}{(1-p-t)^\alpha}\biggr] - \frac
        {(1-p)^{\alpha-1}}{(1-p-t)^\alpha}
    \end{equation}
    and for any $\alpha > 1$ and all $t < 1, \, p \leq 1-t$,
    \begin{equation}
        \frac{\partial^2}{\partial p^2} \psi_\alpha(p,t)= (\alpha-1)\alpha t^2\biggl[\frac{p^{\alpha-2}}{(p+t)^{\alpha+1}}+\frac{(1-p)^{\alpha-2}}{(1-p-t)^{\alpha+1}}\biggr] \geq 0,
    \end{equation}
    from which we can conclude that $\psi_\alpha(p,t)$ is convex in $p$.
        
    First, consider the case $t \geq \nicefrac{1}{\alpha}$. We can bound the first derivative by 
    \begin{equation}
        \frac{\partial}{\partial p}\psi_\alpha(p,t) \geq (\alpha t + p -1)\biggl(\frac{(1-p)^{\alpha-1}}{(1-p-t)^{\alpha}}\biggr) \geq 0,
    \end{equation}
    with equality for $t = \nicefrac{1}{\alpha}$ if $p=0$. Otherwise strict inequality holds. Together with the convexity of $\psi_\alpha$ in $p$, this implies that the maximum value of $\psi_\alpha$ for all $t \geq \nicefrac{1}{\alpha}$ will be achieved at $p = 0$. At this point, we have $\psi_{\alpha}(0,t) = (1-t)^{1-\alpha}$.
    
    Secondly, for $t < \nicefrac{1}{\alpha}$, we first consider $\alpha \geq 2$. Let $P,Q$ be binary probability distributions with $P(0) = p$ and $Q(0)=p+t$ Then, noticing that $x \mapsto x^{\alpha-1}$ is convex for $\alpha \geq 2$, we apply Jensen's inequality to obtain 
    \begin{align*}
        \psi_\alpha(p,t) &= \mathbb E_{X \sim P}\biggl[\biggl(\frac{P(X)}{Q(X)}\biggr)^{\alpha-1}-1\biggr]\\
        &\geq \biggl( \mathbb E_{X \sim P}\biggl[\frac{P(X)}{Q(X)} \biggr] \biggr)^{\alpha-1}-1\\
        &= \biggl(1 + \frac{t^2}{(p+t)(1-p-t)} \biggr)^{\alpha-1}-1\\
        &\geq (1+4t^2)^{\alpha-1}-1. 
    \end{align*}
    Finally, we bound the remaining case with the standard Pinsker's inequality: Since the Rényi-divergence is increasing in $\alpha$ \cite{van2014renyi}, so is $D_{f_\alpha}$. Hence the standard Pinsker's inequlity for the Kullback-Leibler divergence $\text{KL}(P||Q)$ (Rényi-divergence of order $\alpha=1$), $\text{KL}(P||Q) = D_1(P||Q) \geq 2\text{TV}(P||Q)^2$, will continue to hold for $\alpha > 1$. Utilizing \eqref{eq:falphatoalpha}, this yields 
    \begin{equation}
        D_{f_\alpha}(P||Q) \geq \exp\big(2(\alpha-1)\text{TV}(P||Q)^2\big) -1.
    \end{equation} 
    Putting the derived bounds together to form $g_\alpha(t)$ finishes the proof. \qed

\section{Proof of Proposition~\ref{prop:etaisone}}
\label{app:proofetaisone}
For some fixed $P_{Y|X}$, let $P_Y = P_{Y|X}\circ P_X$ and $Q_Y = P_{Y|X}\circ Q_X$. For $\mathcal P = \mathcal P_{\mathcal X}$, we have 
        \begin{align}
        &\eta_\alpha (P_{Y|X},\mathcal P_{\mathcal X}) = \\ &\sup_{P_X,Q_X \in \mathcal P_X: \, D_\alpha (P_X||Q_X) \neq 0} \frac{\log\bigg\{\sum_{y\in \mathcal Y} P_Y(y)^{\alpha}Q_Y(y)^{1-\alpha}\biggr\}}{\log\biggl\{\sum_{x\in\mathcal X}P_X(x)^{\alpha}Q_X(x)^{1-\alpha}\biggr\}}
    \end{align}
    and the statement follows by constructing $P_X$ and $Q_X$ with full support such that in the limiting case to the support set boundaries, we get 

    \begin{equation}
        \frac{\log\biggl\{\sum_{y\in \mathcal Y} P_Y(y)^{\alpha}Q_Y(y)^{1-\alpha}\biggr\}}{\log\biggl\{\sum_{x\in\mathcal X}P_X(x)^{\alpha}Q_X(x)^{1-\alpha}\biggr\}} = 1.
    \end{equation}
    Since the DPI for Rényi-divergences states that $\eta_\alpha(P_{Y|X}) \leq 1$ for all $P_{Y|X}$, this implies $\eta_{\alpha}(P_{Y|X},\mathcal P_{\mathcal X})=1$.

    \textbf{Construction of $P_X$ and $Q_X$:} For ease of notation, let $\mathcal X \coloneqq [N]$. Further, define the indicator function by
    \begin{equation}
        \delta_i(x) = \begin{cases}
            1, &\text{if }x=i\\
            0, &\text{otherwise}.
        \end{cases}
    \end{equation}
    For some very small $\gamma > 0$, let $P_X^\gamma$ be a uniformly perturbed single point mass at $x=i$ for some $i \in [N]$, that is,
    \begin{equation}
        P_X^\gamma(x) = (1-\gamma)\delta_i(x) + \frac{\gamma}{N-1}(1-\delta_i(x)),
    \end{equation}
    and similarly, let
    \begin{equation}
        Q_X^\gamma(x) = \frac{1-\gamma}{2}(\delta_i(x) + \delta_j(x)) + \frac{\gamma}{N-2}(1-\delta_i(x)-\delta_j(x)),
    \end{equation}
    where $j$ is such that any $\mathcal X^{(y)} \coloneqq \{x \in [N]: P_{Y|X=x}(y) > 0\}$ that contains $i$ does not contain $j$. Let $\mathcal S_i \coloneqq \supp(P_{Y|X=i})$. Then we get 
    \begin{align}
        &\lim_{\gamma \to 0} \eta_\alpha(P^\gamma_X,P_{Y|X}) \\&= \frac{\log\biggl\{\sum_{y \in \mathcal S_i}\biggl(P_{Y|X=i}(y)Q^0_X(i)\biggr)^{1-\alpha}\biggl(P_{Y|X=i}(y)\biggr)^\alpha\biggr\}}{\log\biggl\{Q_X^0(i)^{1-\alpha}\biggr\}} \\
        &= \frac{\log\biggl\{Q^0_X(i)^{1-\alpha}\sum_{y\in \mathcal S_i}P_{Y|X=i}(y)\biggr\}}{\log\biggl\{Q^0_X(i)^{1-\alpha}\biggr\}} \\&= \frac{(1-\alpha)\log Q^0_X(i)}{(1-\alpha)\log Q^0_X(i)} = 1,
    \end{align}
    where the second to last equality follows from the row-stochasticity of $P_{Y|X}$.
    By this construction we have $\eta_{\alpha}(P_{Y|X}) \geq 1$. Hence, the data processing inequality for Rényi-divergences implies $\eta_\alpha(P_{Y|X}) = 1$.
    Noticing that choosing $i$ and $j$ in this manner is only possible if there exists some $i,j \in [N]$ such that $\supp(P_{Y|X=i}) \cap \supp(P_{Y|X=j}) = \emptyset$ finishes the proof. \qed

\bibliographystyle{IEEEtranN}
\footnotesize
\bibliography{main}

\end{document}